\documentclass[10pt,conference]{IEEEtran}
\IEEEoverridecommandlockouts
\usepackage{authblk}
\usepackage{amsmath,amssymb,amsfonts,amsthm}
\usepackage[ruled,vlined,linesnumbered]{algorithm2e}
\usepackage{balance}
\usepackage{booktabs}
\usepackage{cite}
\usepackage{color, comment}
\usepackage{graphicx}
\usepackage{lipsum}
\usepackage{mathtools}
\usepackage{multirow}
\usepackage{subfig}
\usepackage{textcomp}
\usepackage{url}
\usepackage[dvipsnames]{xcolor}
\def\BibTeX{{\rm B\kern-.05em{\sc i\kern-.025em b}\kern-.08em
    T\kern-.1667em\lower.7ex\hbox{E}\kern-.125emX}}

\newtheorem{theorem}{Theorem}

\newcommand{\draft}[1]{#1}

\newcommand{\users}{\mathcal{U}}
\newcommand{\model}{\mathcal{M}}
\newcommand{\services}{\mathcal{S}}
\newcommand{\edges}{\mathcal{E}}
\newcommand{\SM}{\mathcal{{SM}}}

\DeclareMathOperator*{\argmax}{arg\,max}

\title{QoS-Aware Placement of Deep Learning Services on the Edge with Multiple Service Implementations}
\author{
    Nathaniel Hudson\textsuperscript{*}, 
    Hana Khamfroush\textsuperscript{*}, and
    Daniel E. Lucani\textsuperscript{\dag}
    \\
    \textsuperscript{*}University of Kentucky, Lexington, KY, USA. Email: nathaniel.hudson@uky.edu, khamfroush@cs.uky.edu
    \\
    \textsuperscript{\dag}Aarhus University, Aarhus, Denmark. Email: daniel.lucani@eng.au.dk
}

\begin{document}
\maketitle

\begin{abstract}
	Mobile edge computing pushes computationally-intensive services closer to the user to provide reduced delay due to physical proximity. This has led many to consider deploying deep learning models on the edge --- commonly known as edge intelligence~(EI). EI services can have many model implementations that provide different QoS. For instance, one model can perform inference faster than another (thus reducing latency) while achieving less accuracy when evaluated. In this paper, we study joint service placement and model scheduling of EI services with the goal to maximize Quality-of-Servcice~(QoS) for end users where EI services have multiple implementations to serve user requests, each with varying costs and QoS benefits. We cast the problem as an integer linear program and prove that it is NP-\emph{hard}. We then prove the objective is equivalent to maximizing a monotone increasing, submodular set function and thus can be solved greedily while maintaining a $(1-1/e)$-approximation guarantee. We then propose two greedy algorithms: one that theoretically guarantees this approximation and another that empirically matches its performance with greater efficiency. Finally, we thoroughly evaluate the proposed algorithm for making placement and scheduling decisions in both synthetic and real-world scenarios against the optimal solution and some baselines. In the real-world case, we consider real machine learning models using the ImageNet 2012 data-set for requests. Our numerical experiments empirically show that our more efficient greedy algorithm is able to approximate the optimal solution with a 0.904 approximation on average, while the next closest baseline achieves a 0.607 approximation on average.
\end{abstract}


\section{Introduction}
\label{sec:intro}
The growth of the Internet has given birth to the advent of the Internet-of-Things~(IoT). This ecosystem consists of countless different devices, or things (e.g., sensors, home appliances), that can seamlessly communicate with one another. More importantly, these devices also serve to generate/collect data. In order to acquire meaningful information from these data, they must first be processed. The scale of the IoT poses a problem for processing these data in a timely fashion through a conventional, centralized approach (e.g., cloud computing). As such, a promising framework to approach this problem has been that of \emph{mobile edge computing}~(MEC)~\cite{ETSI, shi2016edge, mao2017survey, mach2017mobile}. 

The MEC framework considers the deployment of \emph{edge clouds} (or \emph{edge servers}) that provide communication, compute, and storage resources closer to the end user devices to ameliorate latency incurred from physical distance. This physical proximity allows for more immediate and timely data processing for nearby devices in IoT. However, MEC is not without challenges. The hardware resources available at an individual edge cloud pales in comparison to that available at the far away central cloud server. As such, decisions related to how these resources are spent must be optimized. Regardless, MEC remains a promising framework for performing timely data processing for IoT devices, such as smart sensors.

The popularity of \emph{machine learning}~(ML) for performing the task of data processing has skyrocketed in recent years. ML has been shown to be capable of achieving remarkable accuracy for complex tasks (e.g., image classification, video classification, speech-to-text). Due to the flexibility and performance of ML technologies, deploying such models on the edge to process IoT data is promising. Thus, the notion of \emph{edge intelligence}~(EI) has gained prominence. 
A notable feature of EI services, such as image classification, is that several different EI architectures (i.e., deep neural networks) can be implemented to perform inference for some input for a service. These different architectures can have varying trade-offs in terms of the time they take to perform inference, the size of input data they require, and how accurate they are in practice when performing inference. Given these observations, in this work, we adapt the well-studied problem of \emph{service placement} in MEC to consider EI services with different model implementations. Most notably, this problem aims to maximize the Quality-of-Service~(QoS) provided by the edge when services are allowed to have several different model implementations to serve user requests.

\begin{figure}
    \centering
    \includegraphics[width=\linewidth]{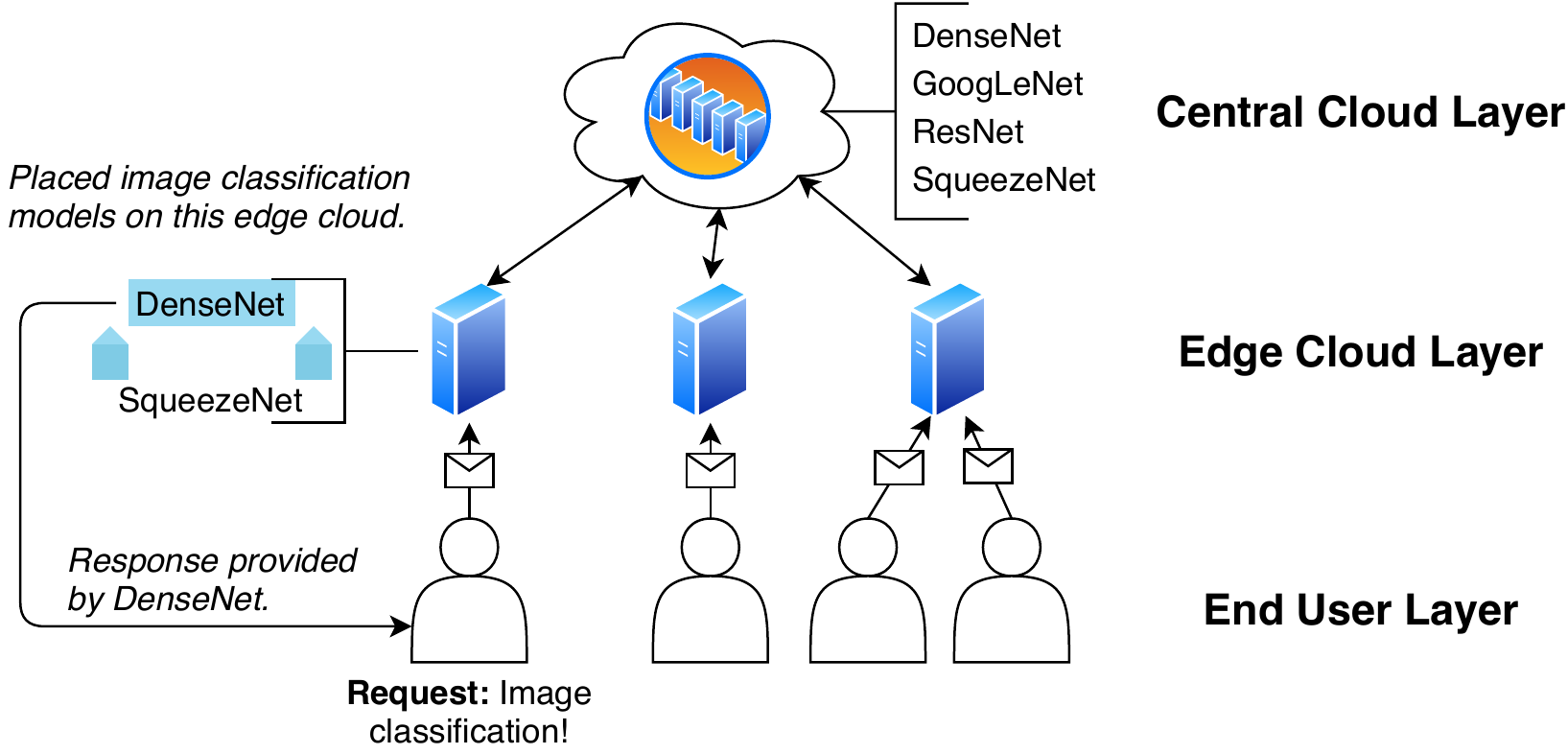}
    \caption{$3$-tier MEC architecture where an edge cloud has two image classification models to serve requests for that service.} 
    \label{fig:MEC-arch}
\end{figure}

To the best of our knowledge, this is the first EI service placement work that considers each service to have multiple implementations. We summarize our contributions as follows:
\begin{itemize}
	\item Introduce the \emph{Placement of Intelligent Edge Services}~(PIES) problem for optimal placement and scheduling of EI services with multiple implementations to maximize QoS provided by the edge.
	\item Prove that the PIES problem is NP-\emph{hard} and that its objective maximizes a monotone increasing, submodular set function under matroid constraints. 
	\item Propose two greedy algorithms: one with a $(1-1/e)$-approximation guarantee and another more efficient greedy algorithm that empirically matches this approximation algorithm with much greater efficiency.
	\item Our empirical results show that both greedy algorithms outperform their minimal approximation guarantee and each achieve an approximation of roughly $0.9$ on average.
\end{itemize}

\section{Related Works}
\label{sec:related_works}
\textbf{Service Placement.}
The problem of deciding which services to place on edge clouds in MEC is commonly known as the \emph{service placement} problem. Many works study the problem with the goal to optimize resource utilization, energy consumption~\cite{wang2014eqvmp, althamary2018popularity}, and Quality-of-Service~(QoS)~\cite{gao2019winning, mahmud2019quality}. Some works consider a static placement decision where the placement decision is made in a single shot~\cite{he2016service, turner2020meeting}. Other works focus on dynamic placement where placement decisions are made over some timespan~\cite{micro_cloud2, ouyang2019adaptive}. The case where edge clouds can share their resources with one another to collaboratively serve user requests has been studied by He et al. in~\cite{he2018s}. The existence of mobile end users has encouraged research studying placement with service migration, where a service processing a request migrates between edge clouds~\cite{ouyang2018follow}. For a recent and comprehensive survey on service placement, please refer to~\cite{salaht2020overview}.
\emph{Quality-of-Service}~(QoS) is very domain specific. Gao et al. in~\cite{gao2019winning} study maximizing QoS, defined as a function of latency, through a joint decision problem for both service placement and network selection where users can be served by more than one edge cloud. Skarlat et al. define FSPP, a QoS-aware service placement problem, as an ILP where QoS is defined as a function of application deadlines --- using the Gurobi solver to provide the optimal solution~\cite{skarlat2017towards}. Yousefpour et al. in~\cite{yousefpour2019fogplan} propose a dynamic fog computing framework, called FogPlan, for dynamic placement and release of services on fog nodes in IoT infrastructure. The authors consider QoS-aware service placement where QoS is considered exclusively w.r.t. delay latency and user's delay tolerance. Wang et al. in~\cite{wang2014eqvmp} study a similar problem wherein they focus on the placement of virtual machines for software-defined data centers to maximize energy efficiency and QoS. An application placement policy using a fuzzy logic-based approach is proposed by Mahmud et al. in~\cite{mahmud2019quality} that maximizes QoE. In this work, the authors consider QoE under a fuzzy framework comprised of sets for access rate, required resources, and processing times. Farhadi et al. in~\cite{khamfroush2021} study service placement and request scheduling on MEC systems for data-intensive applications. They pose the problem of service placement as a set optimization problem and provide an algorithm that demonstrates an approximation bound on optimal solutions.

\textbf{Edge Intelligence.}
The advent of pushing machine learning services to the edge led to the established field of \emph{edge intelligence}~(EI)~\cite{EI_survey2}. Due to the limited resource capacities of MEC environments, a central focus of EI is the design of models that are less costly in terms of resources to run~\cite{EI_survey3}. One proposed solution is to simply prune the elements comprising the deep neural network for the EI service (e.g., remove the number of neurons/units or entire layers)~\cite{prune1, prune2, prune3}. Another proposed idea is to consider a EI model's architecture being \emph{split} across different tiers of the MEC architecture (e.g., one half is run on the edge and the other on the central cloud)~\cite{EI_survey1}. Other works have studied optimizing the QoS provided by deep learning EI models on the edge.  Zhao et al. in~\cite{Xiaobo} study the trade-off between accuracy and latency for offloading decisions regarding deep learning services for provided optimal QoS on the edge through compression techniques. Hosseinzadeh et al. in~\cite{hosseinzadeh2020optimal} study the related problem of offloading and request scheduling in edge systems to maximize QoS for deep learning models under the assumption that placement of these models has been done a priori. There are few works that have formally studied the service placement problem specifically for EI models. Recently in 2021, Zehong, Bi, and Zhang study the problem of optimally placing EI services with the objective of optimizing energy consumption and completion time~\cite{lin2021optimizing}. 

Our work departs from the literature in that we focus on EI service placement where each service can have multiple implementations. To the best of our knowledge from surveying the literature, this is the first work to do so. A similar work by Hung et al. in~\cite{hung2018videoedge} considers multiple implementations of services. However, the focus of that work is on optimal query scheduling of video processing tasks rather than placement of services on edge clouds.

\section{System Model}
\label{sec:model}
\subsection{System Architecture Definition}
\label{sec:architecture}

We consider a 3-tier MEC architecture consisting of a central cloud, edge clouds, and end users (see Figure~\ref{fig:MEC-arch}). The central cloud hosts all model implementations for each service in the environment. However, the objective of this work is to maximize expected QoS provided by the edge clouds and thus we do not closely consider the central cloud. For this work, we focus on three aspects of this environment: edge clouds, end user requests, and available EI service models. For simplicity, we consider requests processed by the cloud to be dropped.

\textbf{Edge Clouds.}
We denote the set of edge clouds $\edges=\{1,\cdots,E\}$ where $E$ is the number of edge clouds. We consider edge clouds to be deployed computational devices that are associated with a wireless access point to connect to users. Each edge cloud is equipped with hardware resources to process the requests provided by end users. Specifically, we consider each edge cloud $e\in\edges$ to have resource capacities $K_e$, $W_e$, and $R_e$ for communication, computation, and storage capacities, respectively. For simplicity, we do not  consider the possibility of edge clouds collaborating to serve user requests. Thus, either a user's request is served by the user's covering edge cloud or is offloaded to the central cloud (dropped).

\textbf{User Requests.}
We denote the set of user requests as $\users=\{1,\cdots,U\}$ where $U$ is the number of user requests. For simplicity, we consider each user to make a single service request. For users that request multiple services, we represent this as separate user requests altogether. Each user is covered by some edge cloud, which will process their service request. We denote the edge cloud covering user~$u$ by $e_u$. Additionally, we denote the set of users an edge cloud~$e$ covers by $\users_e=\{u\in\users : e_u \equiv e\}$. The service some user~$u$ requests is denoted by $s_u$. When submitting a request, users also provide thresholds for accuracy and delay inform the MEC system for how to make decisions w.r.t. QoS. The accuracy threshold provided by user~$u$ is denoted by $\alpha_u\in[0,1]$; the delay threshold provided by user $u$ is denoted by $\delta_u\in[0,\delta_{\max}]$, where $\delta_{\max}$ represents the maximum possible delay. These thresholds are used to prioritize the needs of end users. For instance, some applications that use deep learning are more sensitive to inaccurate answers and others are more time-sensitive. For instance, a self-driving vehicle that uses object detection to detect nearby pedestrians would be more accuracy-sensitive and delay-sensitive than a game on a smartphone.

\textbf{Service Models.}
We consider a set of services that are available for users to request, denoted by $\services=\{1,\cdots,S\}$ where $S$ is the number of available services. 
We assume that there is at least 1 implementation for each service $s\in\services$. However, we also allow for EI services to be implemented by several different machine learning architectures. For brevity, we refer to a single service implementation as a ``service model" for short. The set of implemented service models for service~$s$ is denoted by $\model_s=\{1,\cdots,m_s\}$ where $m_s \geq 1$ is the number of models implemented for service~$s$. For simplicity, we also denote the set of all individual implemented service models by $\SM = \{(s,m) \;:\; \forall s\in\services, m\in\model_{s}\}$. Each service model~$(s,m)\in\SM$ is associated with an accuracy metric~$A_{sm}\in[0,1]$. We assume this value is provided by evaluation using some test data-set (as is typical in machine learning). We denote the expected delay for performing service model~$(s,m)$ for user~$u$ by $D_{sm}(u)$ --- this is defined more explicitly in \S\ref{sec:QoD}. Finally, we denote communication, computation, and storage costs for each service model $(s,m)\in\SM$ by $k_{sm}$, $w_{sm}$, and $r_{sm}$, respectively.

\subsection{Quality-of-Service~(QoS) Definition}
\label{sec:QoS}

We consider QoS for EI models to be comprised of two components: provided model accuracy and incurred delay. As mentioned earlier in \S\ref{sec:architecture}, each user request submitted to the system includes thresholds for requested minimum accuracy, $\alpha_{u}$, and requested maximum delay, $\delta_{u}$. As such, we use these threshold values to compute the expected QoS a service model~$(s,m)$ can provide to user~$u$. The formal definition is provided below in Eq.~\eqref{eq:QoS},

\begin{equation}
	Q(u, s, m) \triangleq 
	\begin{cases}
		\frac{1}{2}[\hat{a}_{sm}(u) + \hat{d}_{sm}(u)] & \text{if } s = s_u
		\\
		0 & \text{otherwise}
	\end{cases}
\label{eq:QoS}
\end{equation}

\noindent 
where $\hat{a}_{sm}(u)$ and $\hat{d}_{sm}(u)$ represent how much service model~$(s,m)$ satisfies user~$u$'s accuracy and delay thresholds, respectively. The summation of these two terms is multiplied by $1/2$ because the maximum possible values for both $\hat{a}_{sm}(\cdot)$ and $\hat{d}_{sm}(\cdot)$ is $1.0$, thus normalizing the range to $Q(\cdot) \in [0,1]$.

\subsubsection{Accuracy Satisfaction} 
\label{sec:QoA}

As stated, each user submits a minimum accuracy threshold, $\alpha_{u}\in[0,1]$, which indicates the amount of accuracy needed to satisfy them. This accuracy of a service model, $A_{sm}$, is a metric retrieved from model evaluation (as is standard with machine learning models). We define this as a nonlinear function in Eq.~\eqref{eq:QoA} below,

\begin{equation}
	\hat{a}_{sm}(u) 
	=
	\begin{cases}
		1 & \text{if } A_{sm} \geq \alpha_u
		\\
		\max(0, 1 - (\alpha_u - A_{sm})) & \text{otherwise}
	\end{cases}
\label{eq:QoA}
\end{equation}

\noindent
where the first case represents when a user's accuracy request has been met and the second case provides reduced satisfaction based on the difference between user-requested accuracy and the evaluated accuracy for service model~$(s,m)$.

\subsubsection{Delay Satisfaction}
\label{sec:QoD}

Similarly, each user submits a maximum delay threshold, $\delta_{u}\in[0,\delta_{\max}]$, to indicate the amount of accuracy they are willing to tolerate. If they receive a response for their request within $\delta_{\max}$ time units, then they are satisfied; otherwise, their satisfaction will degrade. The formal definition is provided below in Eq.~\eqref{eq:QoD},

\begin{equation}
	\hat{d}_{sm}(u) = 
	\begin{cases}
		1 
		& \text{if } D_{sm}(u) \leq \delta_u
		\\
		\max\Big(0, 1 - \frac{D_{sm}(u) - \delta_u}{\delta_{\max}}\Big) 
		& \text{otherwise}
	\end{cases}
\label{eq:QoD}
\end{equation}

\noindent
where $D_{sm}(u)$ is the expected delay from processing user $u$'s request using service model~$(s,m)$. This is defined below in Eq.~\eqref{eq:delay} as the sum of two terms of the transmission delay, $D_{sm}^{tran}(\cdot)$, and the computation delay, $D_{sm}^{comp}(\cdot)$. See below for the formal definition:
\begin{equation}
	D_{sm}(u) = 
	D_{sm}^{tran}(u) + D_{sm}^{comp}(u).
\label{eq:delay}	
\end{equation}

The transmission delay is a function of the communication cost of service model~$(s,m)$ and the communication capacity of user~$u$'s covering edge cloud, $e_{u}$. Additionally, the edge cloud's bandwidth is evenly shared across all of the users it covers --- see Eq.~\eqref{eq:tran-delay},
\begin{equation}
	D_{sm}^{tran}(u) 
	= 
	\frac{k_{sm}}
	     {K_{e_u}/|\users_{e_u}|} 
	=
	\frac{k_{sm} \; |\users_{e_u}|}
	     {K_{e_u}}
\label{eq:tran-delay}
\end{equation}

\noindent 
where $k_{sm}$ is the communication cost for service model $(s,m)$.
Similarly, computation delay is a function of the computation cost of service model~$(s,m)$ and user~$u$'s covering edge cloud, $e_{u}$, and its computation resources --- see Eq.~\eqref{eq:comp-delay},
\begin{equation}
	D_{sm}^{comp}(u) 
	= 
	\frac{w_{sm}}
	     {W_{e_u}/|\users_{e_u}|} 
	=
	\frac{w_{sm} \; |\users_{e_u}|}
	     {W_{e_u}}
\label{eq:comp-delay}
\end{equation}

\noindent
where $w_{sm}$ is the computation cost for service model~$(s,m)$. We assume that the an edge cloud's computation capacity is evenly shared across its covered users.

\section{Problem Definition}
\label{sec:PIES}
Given the system model, we now define the \emph{Placement for Intelligent Edge Services}~(PIES) problem. PIES performs service placement for EI models with multiple implementations with the goal of maximizing QoS, defined in Eq.~\eqref{eq:QoS}. On top of making placement decisions, the PIES problem also decides \emph{which} placed service model will serve a user's request if there are more implementations for a requested service available. For instance, given each user~$u$ makes a request for service~$s_u$, if $u$'s covering edge cloud~$e_u$ has had more than $1$ model of service $s_u$ placed on it, then the PIES problem will also select a placed model to process $u$'s request. 

\subsection{PIES Formulation}
The PIES problem is defined as an \emph{integer linear program}~(ILP) and consists of two types of decisions: \emph{(i)} model placement and \emph{(ii)} model scheduling. For the former, we consider a binary decision variable $\textbf{x}=(x_{e}^{sm})_{\forall e\in\edges, s\in\services, m\in\model_{s}}=1$ if service model~$(s,m)$ is placed on edge cloud~$e$, 0 otherwise. For the latter, we consider another binary decision variable $\textbf{y}=(y_{u}^{m})_{\forall u\in\users, m\in\model_{s_u}}=1$ if user~$u$'s service request is served by its covering edge cloud~$e_u$ with service model~$(s_u,m)$, 0 otherwise. We formally define the PIES problem below:

\begin{subequations}
	\label{eq:PIES}
	\small
	\begin{align}
		\max 
			& \sum_{u \in \users} \sum_{m \in \model_{s_u}} y_u^m Q(u, s_u, m)
		  	  \tag{\ref{eq:PIES}}
		 \\
		 \text{s.t.}
		 	& \sum_{m\in\model_{s_u}} y_u^m \leq 1 
	        & \forall u\in\users
	          \label{eq:PIES-edge_run}
	        \\
		 	& \sum_{s \in \services} \sum_{m \in \model_s} x_e^{sm} r_{sm} \leq R_e
		 	& \forall e \in \edges
		 	  \label{eq:PIES-storage}
		 	\\
		 	& y_u^m \leq x_{e_u}^{s_u m}
		 	& \forall u \in \users, m \in \model_{s_u}
		 	  \label{eq:PIES-y_le_x}
		 	\\
		 	& x_{e}^{sm} \in \{0,1\} 
		 	& \forall e\in\edges, (s,m)\in\SM
		 	\label{eq:PIES-x_binary}
		 	\\
		 	& y_{u}^{m} \in \{0,1\} 
		 	& \forall u \in \users, m \in \model_{s_u}
		 	  \label{eq:PIES-y_binary}
	\end{align}
\end{subequations}

The objective function is defined in Eq.~\eqref{eq:PIES} and maximizes the expected QoS provided by the edge clouds to all users. Constraint~\eqref{eq:PIES-edge_run} ensures that no more than $1$ model is used to process a user's request. Constraint~\eqref{eq:PIES-storage} guarantees that all edge clouds' storage capacities are not exceeded by the summation of the storage costs of their placed service models. Constraint~\eqref{eq:PIES-y_le_x} ensures that users can only be served if their covering edge cloud has placed at least $1$ implementation of their requested service. Finally, constraints~\eqref{eq:PIES-x_binary} and~\eqref{eq:PIES-y_binary} defines $\textbf{x}$ and $\textbf{y}$ as binary decision variables.

\subsection{PIES Problem Complexity \& Properties}
\label{sec:PIES_theory}

Next, we provide proofs related to the hardness of the PIES problem, as well as theoretical properties that can be used to provide approximation guarantees for algorithm design.

\begin{theorem}
    \label{theorem:QoS-np}
    The service model placement sub-problem of the PIES problem is NP-hard.
\end{theorem}

\begin{proof}
    We prove Theorem~\ref{theorem:QoS-np} using a reduction from the \textsc{0/1 Knapsack} problem, which is one of Karp's 21 classical problems proven to be NP-\emph{complete}~\cite{knapsack1}.
    To review, the \textsc{0/1 Knapsack} problem considers $1,\cdots,n$ items with each item having a weight cost $c_i$ and a value $v_i$, as well as a maximum capacity $C$ \draft{for the knapsack}. The problem's objective is to maximize $\sum_{i=1}^{n} v_i x_i$ subject to $\sum_{i=1}^{n} w_i x_i \leq C$ and $x_i\in\{0,1\}$. Here, $x_i=1$ if and only if the $i^{th}$ item is selected to be placed in the knapsack.
    
    We can reduce the 0/1 \textsc{Knapsack} problem to the PIES problem as follows. Suppose that there are $|\services|=n$ available services in the MEC system and only $1$ model per service, i.e., $|\model_{s}|=1\;(\forall s \in \services)$. Each service $i$ has $v_i$ users requesting it --- meaning there are $|\users|=\sum_{i=1}^{n} v_i$ users in total. Suppose there is $|\edges|=1$ edge cloud in the MEC system with storage capacity $R_1=C$, communication capacity $K_1=\infty$, and computation capacity $W_1=\infty$. Let all users be covered by this $1$ edge cloud, such that $e_u=1\;(\forall u \in \users)$. Let the storage costs associated with each service model be $r_{s1}=c_s\;(1 \leq s \leq n)$ and the communication/computation costs $k_{s1}=1,\,w_{s1}=1\;(1 \leq s \leq n)$. We assume that the QoS requirements of all users are relaxed, meaning that $\alpha_u=0.0,\delta_u=\delta_{\max}\;(\forall u\in\users)$. Then we claim that the 0/1 \textsc{Knapsack} problem is feasible if and only if we can maximize expected QoS across all users, i.e., the optimal decision variable of the constructed PIES instance equals $x_{1}^{i1} \equiv x_i$. First, given the optimal solution to the 0/1 \textsc{Knapsack}, placing the services corresponding to the decisions in $x_i$ on the single edge cloud in our constructed instance gives the optimal solution to the PIES problem that maximizes QoS by maximizing the number of user requests served on the edge. Moreover, given an optimal solution to the PIES problem that maximizes QoS across all users, placing the corresponding items in the knapsack  gives an optimal solution to the 0/1 \textsc{Knapsack} problem. Given the decision problem of the 0/1 \textsc{Knapsack} problem is NP-\emph{complete}, it follows that the PIES problem is NP-\emph{hard}. This concludes the proof.
\end{proof}


\begin{theorem}
	Given a service model placement, $\textbf{x}$, the optimal solution to the PIES model scheduling sub-problem is given by a greedy algorithm.
\label{theorem:opt_scheduling}
\end{theorem}

To prove Theorem~\ref{theorem:opt_scheduling}, we consider an auxiliary, undirected multigraph and discuss its construction. Consider a multigraph $\mathcal{G}=(V,L,f)$ where $V$ is the set of nodes, $L \subseteq |V| \times |V|$ the set of links/edges (hereafter we will refer to as \emph{links}), and $f:L \rightarrow \{\{u,v\} : u,v \in V\}$ is the link identifier function. Consider the node set $V$ to be composed of 3 sets of nodes: a set containing a single \emph{root node} $V_R$, a set of \emph{user nodes} $V_U$, and a set of \emph{service nodes} $V_S$. Thus, we say $V = V_R \cup V_U \cup V_S$. We define $V_U$ as the set of users with at least 1 model placed on their edge cloud for their requested service, i.e., $V_U = \{u\in\users : \sum_{m \in M_{s_u}} x_{e_u}^{s_um} > 0\}$. 
%
We define $V_S$ as the set of requested services for each user $u \in V_U$ that have had at least 1 service model placed on their covering edge cloud, i.e., $V_S = \{s_u : \forall u \in V_U\}$. We note that we allow for duplicate services such that each user $u \in V_U$ has its own service node.
%
The set of links, $L$, contains links between node sets $V_R$, $V_U$ and $V_U$, $V_S$. First, there is 1 edge between the root node, $\phi \in V_R$, and each of the user nodes $u \in V_U$, i.e., $\{(\phi, u) : \forall u \in V_U\} \subseteq L$. 
%
Then, a user node $u \in V_U$ can have multiple links to its corresponding service node $s_u \in V_S$. The number of links between $u \in V_U$ and $s_u \in V_S$ is equal to $\sum_{m\in\model_{s_u}} x_{e_u}^{s_um}$, which is the number of model implementations for $s_u$ placed on $u$'s covering edge cloud,~$e_u$. 
Finally, each link is weighted. Links between the root node $\phi \in V_R$ and user nodes $u \in V_U$ are weighted by $0$. Links between each user nodes $u \in V_U$ and its requested service nodes $s_u$ are weighted by the expected $Q(u,s,m)$ where $m$ corresponds with the model type of the respective link between $u$ and $s_u$. A simple example of a multigraph constructed by this method is provided in Fig.~\ref{fig:aux_graph}.

\begin{figure}
	\centering
	\includegraphics[width=0.85\linewidth]{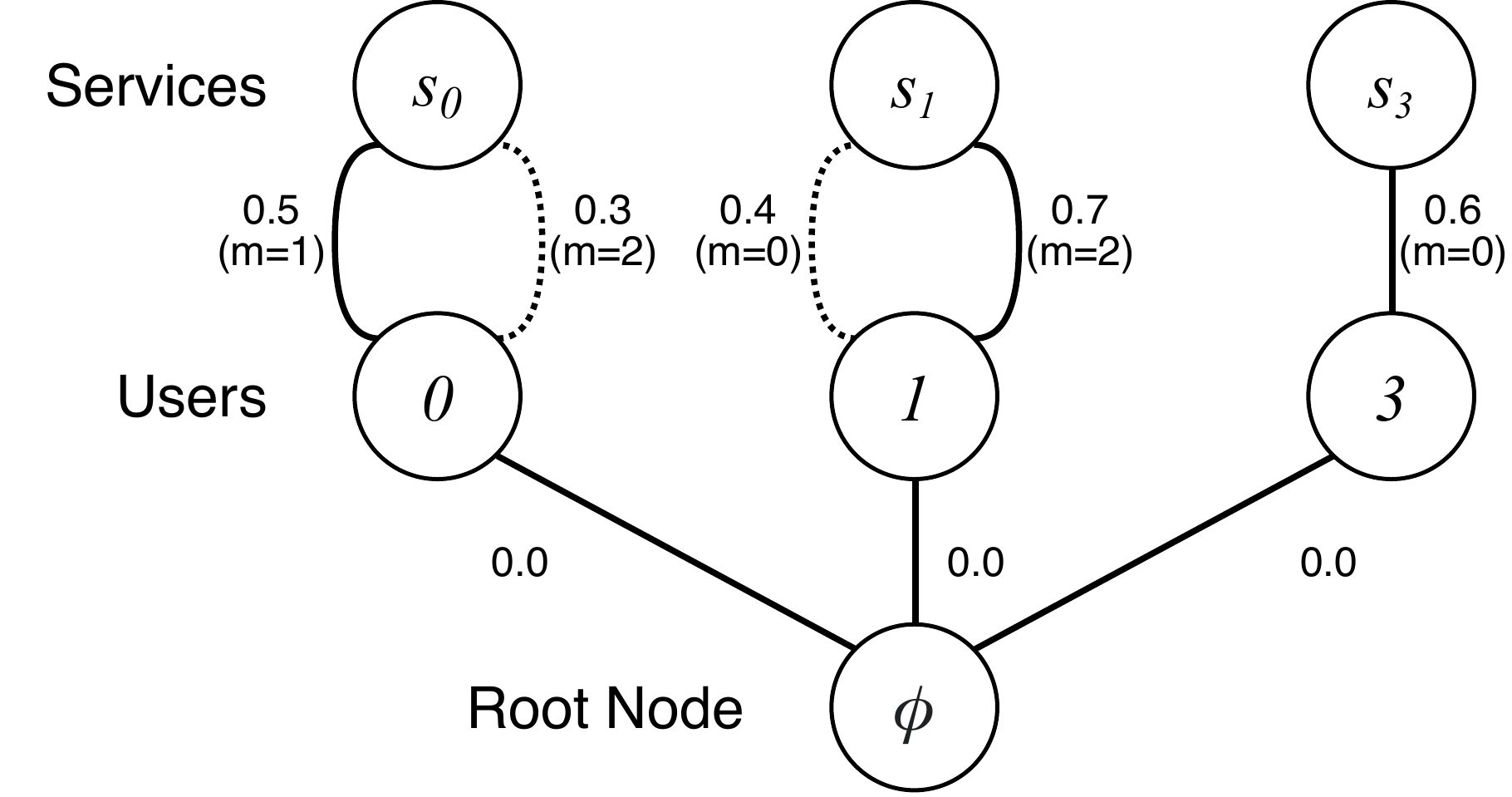}
	\caption{
	    Example auxiliary multigraph representing PIES scheduling sub-problem. Links between user/service nodes correspond to an model implementation for a user's requested service placed on the user's covering edge. Float values correspond with QoS weights. Solid links compose the maximum spanning tree for optimal model scheduling.
	}
	\label{fig:aux_graph}
\end{figure}

\begin{proof}
	Given a service placement decision, $\textbf{x}$, we can prove Theorem~\ref{theorem:opt_scheduling} by constructing an auxiliary multigraph that represents the problem space for the PIES model scheduling sub-problem as described above and shown in Fig.~\ref{fig:aux_graph}. With this graph constructed, it is easy to verify that the optimal solution to the model scheduling sub-problem is equivalent to finding a maximum spanning tree~(MST) --- which can be found by applying a greedy algorithm for minimum spanning trees (e.g., Kruskal's algorithm~\cite{kruskal1956shortest}) and negating all of the edge weights. This concludes the proof.
\end{proof}


\textbf{Proving Submodularity.} 
In order to provide an approximation guarantee for the NP-\emph{hard} PIES placement sub-problem, we prove the PIES service placement sub-problem is maximizing a monotone submodular set function, we rewrite our problem as a set optimization. Let $P(\textbf{x}) \triangleq \{(e,(s,m)) \in \edges \times \SM : x_{e}^{sm} = 1\}$ denote the set of service model placements according to decision variable $\textbf{x}$, where $(e,(s,m))$ means service model~$(s,m)$ is placed on edge cloud~$e$. Next, let $\sigma(P(\textbf{x}))$ denote the optimal objective value of PIES for a given $\textbf{x}$. By writing $P(\textbf{x})$ as $P$, we can rewrite the PIES placement sub-problem as:

\begin{subequations}
	\label{eq:matroid}
	\begin{align}
		\max~
			& \sigma(P)
			\tag{\ref{eq:matroid}}
		\\
		\text{s.t. }	
			& \sum_{(e,(s,m)) \in P \cap P_{e}} r_{sm} \leq R_{e}
			& \forall e \in \edges
			\label{eq:matroid_c1}
			\\
			& P \subseteq \edges \times \SM
			\label{eq:matroid_c2}
	\end{align}
\end{subequations}

\noindent where $P_{e} \triangleq \{e\}\times\{(s,m)\in\SM : r_{sm} \leq R_e\}$ is the set of all possible \emph{single} service model placements at edge cloud~$e$. From here, we observe the following:
\begin{itemize}
	\item \emph{Matroid constraint:} Let $\mathcal{I}$ be the collection of all $P$ satisfying constraints~\eqref{eq:matroid_c1}, \eqref{eq:matroid_c2}. It is then easy to verify that $\mathbb{M}=(\edges\times\SM, \mathcal{I})$ is a matroid. This is known as the \emph{partition matroid} as $\{P_{e}\}_{e\in\edges}$ is a partition of the ground set $(\edges\times\SM)$.
	\item \emph{Monotone submodular objective function:} We show that the objective function~\eqref{eq:matroid} has the following properties.
\end{itemize}

\begin{theorem}
	Function $\sigma(P)$ is monotone increasing and submodular for any feasible $P$.
\label{theorem:submodular}
\end{theorem}

\begin{proof}
	First, adding an element $(e,(s,m))$ to $P$ corresponds to adding new possible users to serve or new links in the auxiliary multigraph $\mathcal{G}$ (see Fig.~\ref{fig:aux_graph}). Under either outcome, the objective value, Eq.~\eqref{eq:PIES}, will either increase or remain unchanged. Thus it is sufficient to say that $\sigma(P)$ is monotone increasing.	
	
	The PIES service model placement sub-problem is submodular iff for every $A,B \subseteq \edges\times\SM$ where $A \subseteq B$ and some $p \notin B$, the following condition holds: $\sigma(A\cup\{p\}) - \sigma(A) \geq \sigma(B\cup\{p\}) - \sigma(B)$. We can define the objective value under optimal scheduling given a set of placements~$P$ as
	\begin{equation}
		\sigma(P) \triangleq \sum_{u\in\users} \sigma_{u}(P)
	\label{eq:sigma}
	\end{equation}
	where $\sigma_{u}(P)$ is the optimal QoS provided to user~$u$ with service model placements~$P$. We can define $\sigma_{u}(P)$ as follows:
	\begin{equation}
		\sigma_{u}(P) \triangleq \max_{(e,(s,m)) \in P} 
		\{
			Q(u,s,m) : e = e_u \wedge s = s_u
		\}
		\cup
		\{
			0
		\}.
	\label{eq:sigma_u}
	\end{equation}
	Next, we claim that for every user $u\in\users$ the following holds: $\sigma_{u}(A\cup\{p\}) - \sigma_{u}(A) \geq \sigma_{u}(B\cup\{p\}) - \sigma_{u}(B)$. We can verify this by elaborating on what these expressions represent and by exploiting their definitions. First, by definition, $\sigma_{u}(A\cup\{p\}) - \sigma_{u}(A)$ represents the increase $\{p\}$ provides to the objective. It should be noted that $\{p\}$ can only provide increase if its provided QoS for $u$ is greater than what $A$ could already provide. Additionally, we note $\sigma_u(A\cup\{p\}) \in \{\sigma_{u}(A), \sigma_{u}(\{p\})\}$. 
	We claim this because $\sigma_u(A\cup\{p\}) \equiv \sigma_{u}(A)$ if $\sigma_{u}(A) \geq \sigma_{u}(\{p\})$, meaning an already placed model in $A$ is scheduled to serve $u$'s request, or $\sigma_u(A\cup\{p\}) \equiv \sigma_{u}(\{p\})$ if $\sigma_{u}(A) < \sigma_{u}(\{p\})$, meaning the new service model placed by the new placement~$\{p\}$ is scheduled to serve $u$'s request. 
	These observations similarly hold for $B$ as well. Now, we prove the following inequality $\sigma_{u}(A\cup\{p\}) - \sigma_{u}(A) \geq \sigma_{u}(B\cup\{p\}) - \sigma_{u}(B)$ directly by its possible cases:\footnote{Note the case that $\sigma_{u}(A\cup\{p\}) \equiv \sigma_{u}(A)$ and $\sigma_{u}(B\cup\{p\}) > \sigma_{u}(B)$ can never occur. This is due to the fact that $A \subseteq B$ and thus if $\{p\}$ provides greater QoS for user~$u$ than $B$, then it follows that is also true for $A$.}
	\begin{itemize}
		\item \textsc{Case 1.} $\sigma_{u}(A\cup\{p\}) \equiv \sigma_{u}(A)$ and $\sigma_{u}(B\cup\{p\}) \equiv \sigma_{u}(B)$.
				It is easy to verify that $\sigma_{u}(A\cup\{p\}) - \sigma_{u}(A) \geq \sigma_{u}(B\cup\{p\}) - \sigma_{u}(B)$ becomes $0 \geq 0$, which holds.
		\item \textsc{Case 2.} $\sigma_{u}(A\cup\{p\}) > \sigma_{u}(A)$ and $\sigma_{u}(B\cup\{p\}) \equiv \sigma_{u}(B)$.
				It is easy to verify that the lefthand side of the original inequality becomes a value $>0$ and the righthand side becomes $0$, thus the inequality holds.
		\item \textsc{Case 3.} $\sigma_{u}(A\cup\{p\}) > \sigma_{u}(A)$ and $\sigma_{u}(B\cup\{p\}) > \sigma_{u}(B)$. 
				Since $A \subseteq B$, any service model placed under $A$ is also placed under $B$. Intuitively, any increase to QoS for user~$u$ provided by $A\cup\{p\}$ can be matched by $B\cup\{p\}$ because $B$ has every service model to serve $u$'s request that $A$ has. Additionally, $B$ could also have service models that provide greater QoS for $u$ than $A$ due to it having more service models placed. Thus, it must follow that the original inequality holds because $A$'s increase in QoS for $u$ is always at least as large as $B$'s increase.
	\end{itemize}

	\noindent 
	Thus, $\sigma_{u}(\cdot)$ is submodular for every $u\in\users$. Since any function that is a summation of submodular functions is also submodular~\cite{submodular_seminal}, it then follows that $\sigma(\cdot)$ is submodular. This concludes the proof.
\end{proof}

\section{Efficient Algorithm Design}
\label{sec:algorithm}

\begin{algorithm}[t]
	\footnotesize
    \SetKwData{Left}{left}\SetKwData{This}{this}\SetKwData{Up}{up}
    \SetKwInOut{Input}{Input}\SetKwInOut{Output}{Output}
    
    \Input{Service placement $(\textbf{x})$, Input parameters of~\eqref{eq:PIES}}
    \Output{Model scheduling $(\textbf{y})$}
	
	Initialize $\textbf{y} \gets (y_u^m=0)_{\forall u\in\users, m\in\model}$\;
    \ForEach{user $u\in\users$}{
    	\If{$\sum_{m\in\model_{s_u}} x_{e_u}^{s_um} > 0$}{
    		$m^* \gets \argmax\limits_{m\in\model_{s_u}} \{ Q(u, s_u, m), \forall u\in\users_e : x_{e_u}^{s_um} = 1\}$\;
    		$y_u^{m^*} \gets 1$\;
    	}
    }
    
    \Return \textbf{y}\;

\caption{Optimal Model Scheduling~(OMS)}
\label{alg:greedy-scheduling}
\end{algorithm}

\subsection{Scheduling Sub-Problem}
\label{sec:scheduling_alg}
In \S\ref{sec:PIES_theory}, we proved the PIES model scheduling sub-problem can be optimally solved with a greedy solution (see Theorem~\ref{theorem:opt_scheduling}). As such, we introduce a simple greedy algorithm, \emph{Optimal Model Scheduling}~(OMS). The pseudocode is provided in Algorithm~\ref{alg:greedy-scheduling}. OMS works in a straightforward fashion. Given a service placement decision and the PIES input parameters, OMS iterates through each user $u\in\users$ and if there is at least one model of their requested service, $s_u$, on their covering edge cloud, $e_u$, then OMS selects the model that provides the greatest QoS to user~$u$. The runtime for OMS is $O(|\users||\model_{s}^{\max}|)$ where $|\model_{s}^{\max}|=\max_{s\in\services}(|M_{s}|)$.

\subsection{Placement Sub-Problem}
\label{sec:placement_alg}

Here, we introduce two algorithms that can be used to solve the service model placement sub-problem for PIES. The first is an approximation algorithm that exploits the theoretical properties of the PIES objective (discussed in \S\ref{sec:PIES_theory}) to achieve an approximately optimal solution. The latter algorithm mimics some of the logic of this algorithm while reducing computational heft.

\subsubsection{Approximation Algorithm}
\label{sec:AGP}

Because the PIES problem aims to maximize a monotone increasing, submodular set function under matroid constraints (see Theorem~\ref{theorem:submodular}), a standard greedy algorithm can provide a $(1-1/e)$-approximation of the optimal solution~\cite{calinescu2007maximizing}. As such, we introduce \emph{Approximate Greedy Placement}~(AGP) which serves as an approximation algorithm for the PIES placement sub-problem. Its pseudocode is provided in Algorithm~\ref{alg:simple_greedy_placement}. AGP iterates through each edge cloud $e\in\edges$ and, in each iteration, it finds the set of service models $(s,m) \in \SM$ that can be placed on $e$ without violating the storage capacity constraint. It then computes the objective value using optimal model scheduling via Eq.~\eqref{eq:sigma} to find the immediate best choice. Once there are no more legitimate choices to choose from, it moves on to the next edge cloud. Once iteration through edge clouds is finished, it converts the placement decisions (represented as a set) into the standard format for the decision variable. However, AGP's runtime is not desirable due to its need to compute optimal scheduling for each possible option at each iteration. 


\begin{algorithm}[t]
	\footnotesize
	\SetKwData{Left}{left}\SetKwData{This}{this}\SetKwData{Up}{up}
    \SetKwInOut{Input}{Input}\SetKwInOut{Output}{Output}
	
	\Input{Input parameters of~\eqref{eq:PIES}}
	\Output{Service placement $(\textbf{x})$}
	
	Initialize $\textbf{x} \gets (x_e^{sm}=0)_{\forall e\in\edges, s\in\services, m\in\model}$\;
	$P \gets \{\}$    \tcp*{Placement decisions $\forall e \in \edges$.} 
	\ForEach{$e \in \edges$}{
		$\hat{P} \gets \{\}$    \tcp*{Placement decisions for this $e$.} 
		$\hat{R} \gets R_{e}$\;
		\Repeat{$|L \setminus \hat{P}| = 0$}{
			$L \gets \{(s, m) \in \SM \setminus \hat{P} : r_{sm} \leq \hat{R}\}$\;
			$s^{*}, m^{*} \gets \argmax\limits_{(s,m) \in L} \sigma(P\cup \{(e, (s, m))\})$\;
			$\hat{P} \gets (s^*,m^*)$\;
			$P \gets P \cup \{(e, (s^{*}, m^{*}))\}$\;
			$\hat{R} \gets \hat{R} - r_{s^{*}m^{*}}$\;
		}
	}
	\ForEach{$(e, (s, m)) \in P$}{
		$x_{e}^{sm} \gets 1$\;
	}
	\Return $\textbf{x}$\;
	
\caption{Approx. Greedy Placement (AGP)}
\label{alg:simple_greedy_placement}
\end{algorithm}
\begin{algorithm}[t]
	\footnotesize
    \SetKwData{Left}{left}\SetKwData{This}{this}\SetKwData{Up}{up}
    \SetKwInOut{Input}{Input}\SetKwInOut{Output}{Output}
    
    \Input{Input parameters of~\eqref{eq:PIES}}
    \Output{Service placement $(\textbf{x})$}
    
	Initialize $\textbf{x} \gets (x_e^{sm}=0)_{\forall e\in\edges, s\in\services, m\in\model}$\;
    \ForEach{edge $e\in\edges$}{
    
    	Initialize hash-map $\mathbf{v}$ with default values of $0.0$\;
    	\ForEach{user $u \in \users_{e}$}{
    	    \ForEach{model type $m \in \model_{s_u}$}{
    	        $\mathbf{v}_{s_u m} \gets \mathbf{v}_{s_u m} + Q(u,s_u,m)$\;
    	    }
    	}
    	
    	$\mathcal{A} \gets \{\}$   \tcp*{Considered service models.} 
    	$\mathcal{B} \gets \{\}$   \tcp*{Satisfied users.}
    	$\hat{R} \gets R_{e}$      \tcp*{Remaining storage.}
    	
		\Repeat{$(\hat{R} = 0)  \vee  (|\users_{e}| = |\mathcal{B}|)  \vee  (\mathcal{|A|} = |\text{keys}(\mathbf{v})|)$}{
			$s^*, m^* \gets \argmax\limits_{(s,m) \in \text{keys}(\mathbf{v}) \setminus \mathcal{A}} \{\mathbf{v}_{sm}\}$\;
			\If{$r_{s^*m^*} \leq \hat{R}$}{
				$x_{e}^{s^*m^*} \gets 1$\;
				$\hat{R} \gets \hat{R} - r_{s^*m^*}$\;
				\ForEach{$m \in \model_{s^*}$ where $(s^*,m) \notin \mathcal{A}$}{
					$\mathbf{v}_{sm} \gets \sum\limits_{u \in \users_e \setminus \mathcal{B}} Q(u,s^*,m) - Q(u,s^*,m^*)$\;
				}
			}
			
			$\mathcal{A} \gets \mathcal{A} \cup \{(s^*, m^*)\}$\;
			
			\ForEach{$u \in \users_e$}{
				\lIf{$Q(u,s^*,m^*) = 1$}{
					$\mathcal{B} \gets \mathcal{B} \cup \{u\}$
				}
			}
		}
    }    
    \Return \textbf{x}\;

\caption{Efficient Greedy Placement~(EGP)}
\label{alg:EGP}
\end{algorithm}

\subsubsection{Efficient Algorithm}
\label{sec:EGP}

Due to the heavy runtime complexity of AGP, there is a need for a more efficient algorithm that can relatively match AGP's performance w.r.t. approximating the optimal solution without the large computational cost. Thus, we introduce the \emph{Efficient Greedy Placement}~(EGP) algorithm. EGP iteratively places models by keeping a record of anticipated benefit of placing any given service model on an edge cloud. It does this without computing optimal scheduling, thus reducing its computational cost. EGP's pseudocode is provided in Algorithm~\ref{alg:EGP}. In line 1, the placement decision variable is initialized. Then, on line 2, we begin to iterate through each edge cloud to decide which service models should be placed on the current edge cloud. Line 3 initializes an empty hash-map and lines 4-6 compute the total QoS each service model relevant for the current edge cloud can provide towards the objective function. This data structure will be updated as decisions are made. Lines 7-9 initialize some supporting variables for EGP's logic. $\mathcal{A}$ records service models that have been considered for placement at some point for the current edge cloud; $\mathcal{B}$ keeps track of the set of users who can be provided maximum QoS; and $\hat{R}$ tracks the remaining storage capacity. Lines 10-20 find the service model $(s^*,m^*)$ that provide the maximum QoS among the service models we have yet to consider from our hash-map. If $(s^*,m^*)$ can be placed without violating the storage constraint, then it will be placed (line 13) and remaining storage will be reduced (line 14). Then, in lines 15-16 we recalculate the benefit of each other model implementation for $s^*$ by computing $\sum_{u\in\users_e\setminus\mathcal{B}} Q(u,s^*m)-Q(u,s^*,m^*)$. Since the newly placed model degrades the benefit from picking other implementations of the same service, we sum the difference between these other models and the newly placed model to reevaluate the benefit of placing them. 
Lines 10-20 repeat until either there is no more storage space, all of the edge cloud's user's achieve maximum QoS, or we have considered all models relevant for the edge cloud (according to the keys recorded in the hash-map). The runtime complexity of EGP is $O(|\users|+|\users||\model_s^{\max}|)$ where $|\model_{s}^{\max}|=\max_{s\in\services}(|M_{s}|)$.

\section{Experimental Design \& Results}
\label{sec:experiments}
We consider numerical simulations and a real-world implementation with real image data and ML models. Algorithms are implemented in Python 3.8 and experiments are largely run on a macOS $64$~bit machine with a $3.2$~GHz quad-core Intel Core i5 processor and $32$~GB $1600$~MHz DDR3 memory.

\subsection{Baselines}
We compare AGP and EGP to the ILP defined in Eq.~\eqref{eq:PIES} which is solved using the PuLP Python library~\cite{PuLP} and the CBC solver~\cite{CBC} (referred to as ``OPT"). We also adapt the standard dynamic programming algorithm for the \textsc{0/1 Knapsack} problem for the PIES problem (referred to as ``SCK"). SCK considers the individual service models as the separate items --- with their storage costs serving as their weights and Eq.~\eqref{eq:QoS} as their values. SCK will use Alg.~\ref{alg:greedy-scheduling} for scheduling. We also consider a random placement and scheduling heuristic (referred to as ``RND") as our other baseline.

%

\begin{figure}
    \centering
    \includegraphics[width=\linewidth]{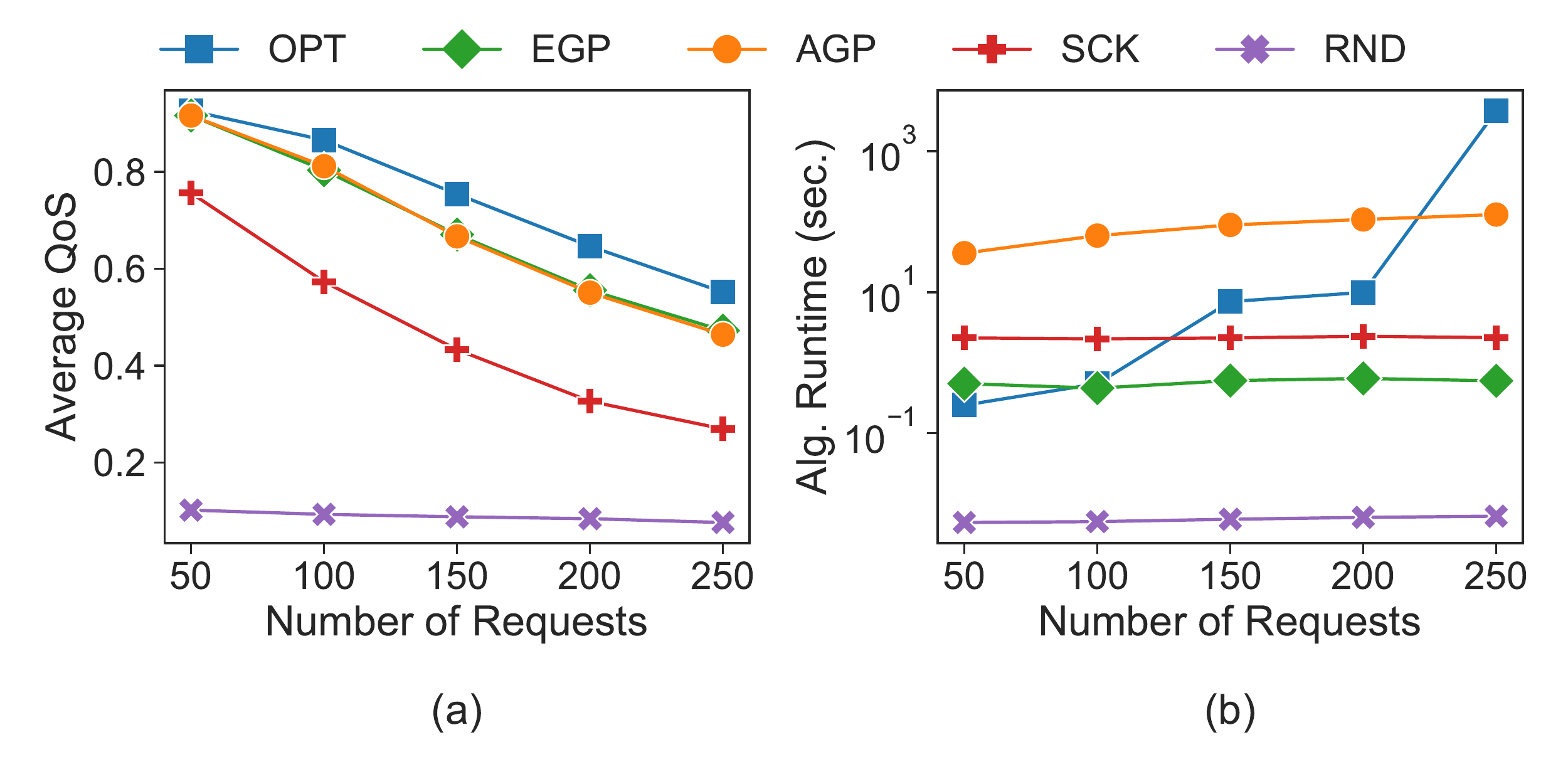}
    \caption{\textbf{Validation Test.} For these experiments, we consider $|\users| = [50, 100, 150, 200, 250]$, with $10$ trials each. We perform this validation test to confirm the efficacy of EGP relative to the optimal solution and the approximation algorithm, AGP.}
    \label{fig:validation_test}
\end{figure}

\begin{figure}
    \centering
    \includegraphics[width=\linewidth]{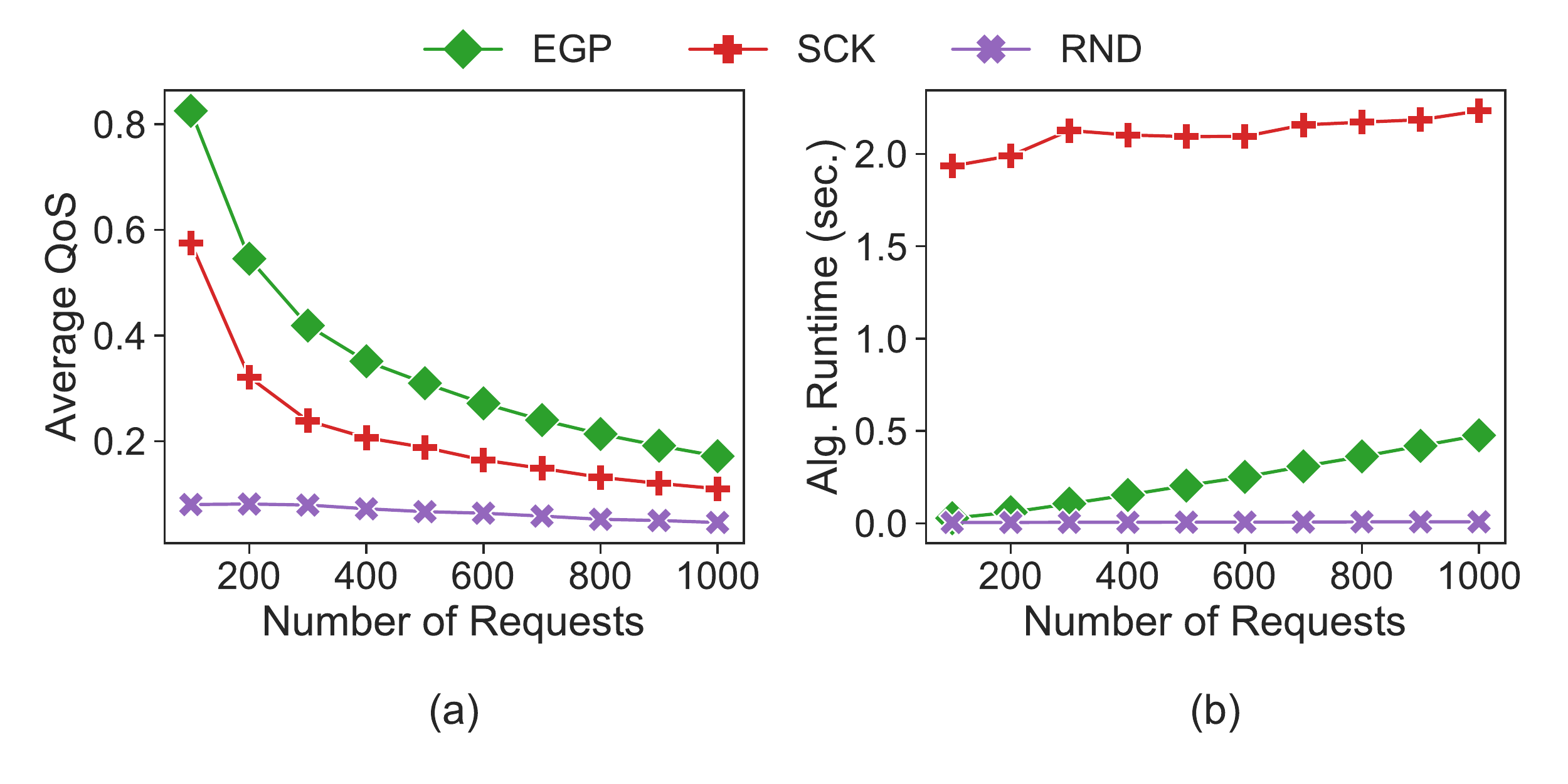}
    \caption{\textbf{Numerical Results.} Experiments with more user requests $|\users| = [100, 200, \cdots, 1000]$, with $100$ trials each.}
    \label{fig:numerical_test}
\end{figure}

\subsection{Numerical Simulations}
We sample uniformly random integer values for edge capacities using the following distributions $K_{e}, W_{e}\in[300,600]$, and $R_{e}\in[100,200]\;(\forall e \in \edges)$. Service model storage costs are similarly uniformly sampled integer values where $k_{sm}, w_{sm}\in[15,30]$ and $r_{sm}\in[10,20]\;(\forall (s,m) \in \SM)$. Service models' cached accuracy, $A_{sm}$, are sampled from a Gaussian distribution with a mean of $0.65$ and a standard deviation of $0.1$ (sampled values are clipped to the range $[0,1]$). We assume users request service types from $\services$ uniformly at random. User accuracy thresholds, $\alpha_u$, are set to $1-\epsilon$ where $\epsilon$ is sampled from an exponential distribution clipped to the range $[0,1]$ with $\lambda=0.125$. User delay thresholds, $\delta_u$, are set to a sampled value (clipped to range $[0,\delta_{\max}]$ where $\delta_{\max}=10$ seconds) from an exponential distribution where $\lambda=1.5$.  Finally, we consider $|\edges|=10$ edge clouds, $|\services|=100$ services with each service having a random number of implementations in the range of $[1,10]$ (sampled uniformly). We increase the number of users for experiments.

First, we compare our proposed algorithms (EGP and AGP) to the optimal solution. Due to the hardness of the PIES problem, we  consider a validation case to demonstrate EGP performance relative to the optimal solution\footnote{It is worth noting in some larger scenarios we ran, it took over $20$~hours for the optimal solver to complete.}  and AGP. In Fig.~\ref{fig:validation_test}a, we see that the AGP and EGP algorithms perform well for approximating the optimal solution provided by the solver. More importantly, we find that EGP is able to match AGP's performance even though it has a proven approximation bound. Specifically, we find that, on average, AGP and EGP achieve an approximation ratio of $0.900$ and $0.904$, respectively.\footnote{These values are computed by taking the QoS for each algorithm in one experiment trial and dividing it by the QoS provided by the optimal solution.} In Fig.~\ref{fig:validation_test}b, we see that EGP greatly outclasses both Optimal and AGP in terms of efficiency. The excessive cost associated with AGP's runtime is due to its reliance on performing optimal model scheduling for each candidate service model at each selection step. EGP also manages to best SCK. When considering more requests, we see in Fig.~\ref{fig:numerical_test} that our EGP solution achieves roughly $50\%$ more QoS than SCK while still managing to be more efficient.

\subsection{Real-World Implementation}
We consider a simple real-world setup using set of $2$ Nvidia Jetson Nano and $1$ Raspberry Pi 3B+ nodes as IoT devices and a 2013 Apple iMac serving as the edge cloud. Each IoT device hosts roughly a third of the 2012 ImageNet dataset~\cite{imagenet2012} via non-overlapping subsets. Each IoT device submits $100$ requests with randomly sampled images from these data. Requests are submitted wirelessly for image classification service models hosted on the edge cloud. Here we focus on the multi-implementation aspect of the PIES placement sub-problem where multiple implementations for image classification can be placed and how this affects the QoS in the real-world case. 
Using PyTorch~\cite{pytorch}, we evaluate pre-trained image classification models to record their accuracy metric over the ImageNet 2012 data and record the average time needed for each model to perform inference, see Table~\ref{tab:img_classification_models}. 
The edge cloud can place $|R_{e}|=1$ model where each ML model is associated with $r_{sm}=1$ storage cost. Since all the models accept the same data size, we fix the communication costs $w_{sm}=1$ for all the ML models. The communication and computation capacities for the edge cloud are robustly tuned to match the real-world computation and communication delay.  Similarly to the numerical simulation setup, each request's $\alpha_u$ is sampled from $1-\epsilon$ where $\epsilon$ is sampled from an exponential distribution clipped to the range $[0.0,1.0]$ with a rate parameter $\lambda=0.0625$. Each sample's delay threshold $\delta_u$ is sampled from a Gaussian distribution with a mean of $0.5$ and a standard deviation of $0.125$, clipped to the range $[0,\delta_{\max}]$ where $\delta_{\max}=1.0$ second. The QoS for each request is calculated using Eq.~\eqref{eq:QoS} using the real-time incurred latency (in seconds) and the evaluated model accuracy. 

In Fig.~\ref{fig:realworld_results}a, we see that all considered algorithms but random are able to match the optimal solution --- with all non-random algorithms exclusively placing MobileNet in Fig.~\ref{fig:realworld_results}b. In Fig.~\ref{fig:realworld_results}a, the QoS distribution of the non-random algorithms are much more concentrated on the upper end when compared to random. In this setup, random does better than in Figs.~\ref{fig:validation_test} ,\ref{fig:numerical_test} because a request will never be dropped (i.e., there is always an image classification available to provide \emph{some} QoS). Thus, future real-world implementations must consider various service types. For the meantime, these results show promise but could be improved and expanded upon by considering a more robust real-world setup with more EI service types (e.g., speech-to-text, video classification).


\begin{table}
	\centering
	\caption{Image classifications models used for the real-world implementation with model accuracy metrics and average computational delay from evaluation with ImageNet 2012 data.}
	\label{tab:img_classification_models}
	\begin{tabular}{rcc}
		\toprule
		Models & Accuracy, $A_{sm}$ & Avg. Comp. Delay (sec.) \\
		\midrule
		AlexNet~\cite{alexnet} & 56.52\% & 0.04
		\\
		DenseNet~\cite{densenet} & 77.14\% & 0.47
		\\
		GoogLeNet~\cite{googlenet} & 69.78\% & 0.13
		\\
		MobileNet~\cite{mobilenet} & 71.88\% & 0.06
		\\
		ResNet~\cite{resnet} & 69.76\% & 0.08
		\\
		SqueezeNet~\cite{squeezenet} & 58.09\% & 0.07
		\\
		\bottomrule
	\end{tabular}	
\end{table}


\begin{figure}
    \centering
    \includegraphics[width=\linewidth]{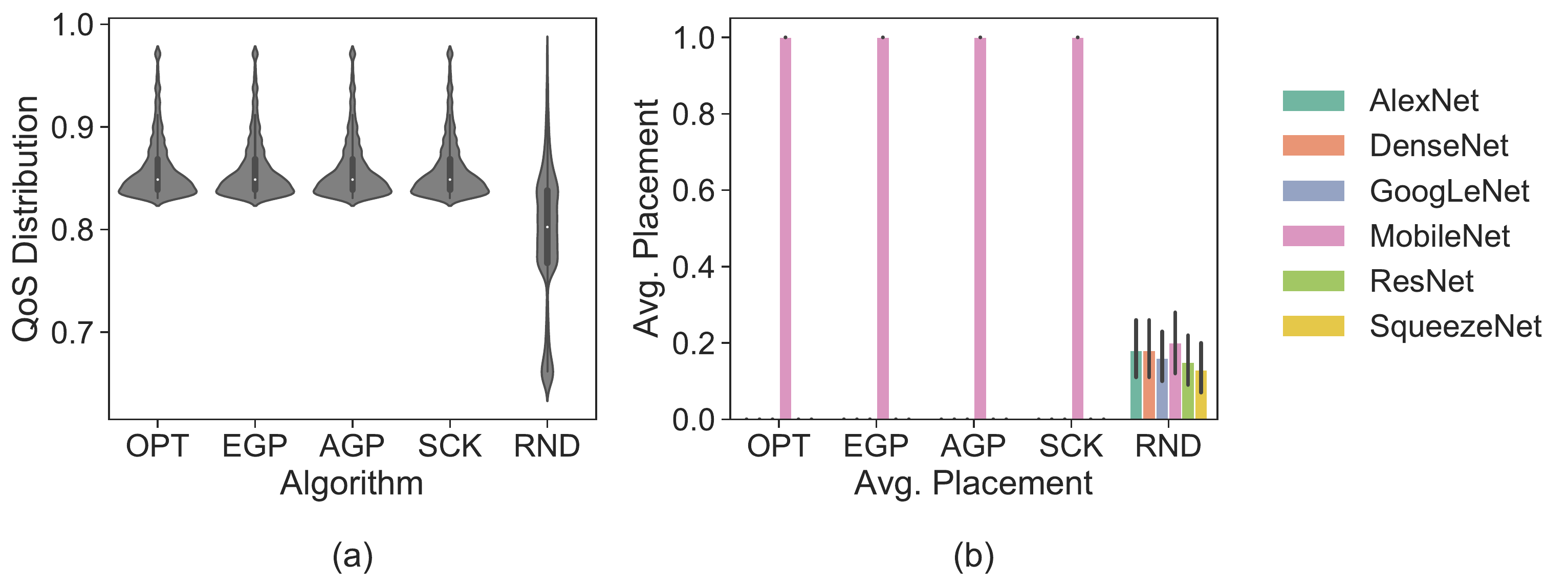}
    \caption{\textbf{Real world implementation.} (5a) QoS distribution achieved by each of the algorithms. (5b) Average placement decision for each image classification model across $100$ trials.}
    \label{fig:realworld_results}
\end{figure}

\section{Conclusions}
\label{sec:conclusion}
The PIES problem, to the best of our knowledge, is the first service placement and scheduling problem that explicitly considers the case of EI services having multiple implementations available for the same service. We proved that the PIES problem is NP-\emph{hard} and prove a greedy set optimization algorithm can provide a $(1-1/e)$-approximation guarantee of the optimal solution. We then introduce a streamlined greedy algorithm that empirically matches this algorithm's performance with much greater efficiency. While these results are preliminary, they serve as a foundational first step towards this breed of service placement. For future work, we plan to consider more dynamic extension of this work where service placement decisions are made over a time horizon rather than all at once. Additionally, we will expand the real-world setup to include more EI services (e.g., video classification) with multiple implementations for placement and scheduling.

\section*{Acknowledgements}
This material is based upon work supported by the National Science Foundation under grant no.~CSR-1948387. This work was also partially funded by Cisco Systems Inc. under the research grant number 1215519250. We thank both our sponsors for their generous support.


\balance

\bibliographystyle{ieeetr}

\end{document}